\documentclass[sigconf,authorversion]{acmart}

\usepackage{amsthm}
\usepackage{mathtools}
\usepackage{cleveref}
\usepackage{float}
\usepackage{tabularx,multirow}

\theoremstyle{definition}
\newtheorem{definition}{Definition}[section]
\newtheorem{theorem}{Theorem}[section]
\newtheorem{lemma}[theorem]{Lemma}

\AtBeginDocument{%
  }

\copyrightyear{2023}
\acmYear{2023}
\setcopyright{acmlicensed}\acmConference[ICPP 2023]{52nd International Conference on Parallel Processing}{August 7--10, 2023}{Salt Lake City, UT, USA}
\acmBooktitle{52nd International Conference on Parallel Processing (ICPP 2023), August 7--10, 2023, Salt Lake City, UT, USA}
\acmPrice{15.00}
\acmDOI{10.1145/3605573.3605588}
\acmISBN{979-8-4007-0843-5/23/08}

\begin{document}

\title{Recoil: Parallel rANS Decoding with Decoder-Adaptive Scalability}

\author{Fangzheng Lin}
\email{lin_toto@toki.waseda.jp}
\author{Kasidis Arunruangsirilert}
\email{kasidis@fuji.waseda.jp}
\affiliation{
  \institution{Waseda University}
  \streetaddress{3-4-1 Okubo}
  \city{Shinjuku}
  \state{Tokyo}
  \country{JP}
  \postcode{169-8555}
}

\author{Heming Sun}
\email{sun-heming-vg@ynu.ac.jp}
\affiliation{
  \institution{Yokohama National University}
  \streetaddress{79-1 Tokiwadai}
  \city{Hodogaya}
  \state{Yokohama}
  \country{JP}
  \postcode{240-0067}
}

\author{Jiro Katto}
\email{katto@waseda.jp}
\affiliation{
  \institution{Waseda University}
  \streetaddress{3-4-1 Okubo}
  \city{Shinjuku}
  \state{Tokyo}
  \country{JP}
  \postcode{169-8555}
}

\begin{abstract}
Entropy coding is essential to data compression, image and video coding, etc. The Range variant of Asymmetric Numeral Systems (rANS) is a modern entropy coder, featuring superior speed and compression rate. As rANS is not designed for parallel execution, the conventional approach to parallel rANS partitions the input symbol sequence and encodes partitions with independent codecs, and more partitions bring extra overhead. This approach is found in state-of-the-art implementations such as DietGPU. It is unsuitable for content-delivery applications, as the parallelism is wasted if the decoder cannot decode all the partitions in parallel, but all the overhead is still transferred. 

To solve this, we propose Recoil, a parallel rANS decoding approach with decoder-adaptive scalability. We discover that a single rANS-encoded bitstream can be decoded from any arbitrary position if the intermediate states are known. After renormalization, these states also have a smaller upper bound, which can be stored efficiently. We then split the encoded bitstream using a heuristic to evenly distribute the workload, and store the intermediate states and corresponding symbol indices as metadata. The splits can then be combined simply by eliminating extra metadata entries. 

The main contribution of Recoil is reducing unnecessary data transfer by adaptively scaling parallelism overhead to match the decoder capability. The experiments show that Recoil decoding throughput is comparable to the conventional approach, scaling massively on CPUs and GPUs and greatly outperforming various other ANS-based codecs.
\end{abstract}




\maketitle

\section{Introduction}
The demand for high-quality entertainment content, such as high-resolution images, Ultra-High-Definition (UHD) 4K and 8K videos, and the streaming of VR and AR content, is rapidly growing. However, the communication bandwidth, especially via wireless links \cite{10118777}, is very limited. Therefore, in these applications, data compression always plays a crucial role in both user experience enhancement and cost saving, by reducing the amount of transmission bandwidth and storage. Data compression is generally achieved by Entropy Coding, which efficiently encodes symbols close to their Shannon limit. The Asymmetric Numeral Systems (ANS) \cite{ANS} family is a series of modern entropy coders. Its variants are widely used in recent compression and coding algorithms such as JPEG-XL \cite{JPEGXL}, FSE \cite{FSE}, Zstandard \cite{Zstd}, etc., providing both superior compression and decompression speed, and high compression rate.

To achieve higher throughput and lower latency in these applications, exploiting parallelism out of entropy coding algorithms has always been desirable. However, this is not simple to achieve: entropy coding inherently relies on variable-length codes, resulting in internal data dependencies, which could be complex to split apart. In particular, an ANS bitstream must be serially encoded from the first symbol to the last, and decoded from the back to the front to get the symbols back in reverse order. Therefore, the conventional approach to parallel entropy coders often involves partitioning the input symbol sequence into smaller sub-sequences, and encoding each sub-sequence with an independent entropy coder. This approach is seen in state-of-the-art massively parallel rANS codec implementations such as DietGPU \cite{DietGPU}. This method presents a trade-off between compression rate and the level of parallelism: as the input sequence is partitioned into more sub-sequences, the worsening of the compression rate becomes more dominant, due to the almost linearly increasing amount of coding overhead.

In content-delivery applications, the server must account for the different parallel capacities of clients (decoders) by encoding the content with more sub-sequences. While this may allow massive parallelism, it also makes the file size larger. Besides, a decoding machine with a state-of-the-art GPU may be able to decode tens of thousands of sub-sequences in parallel, while a budget CPU can only decode a few at once. Therefore, the budget CPU couldn't take advantage of the massive-parallelism optimized content, but still needs to receive all the data required to exploit the parallelism intended for high-end machines. On the other hand, the server could also prepare multiple variations of the content to contend for different parallel capacities, but this is still undesirable as it creates great storage and computational overhead for the server. This is because of the primary drawback with the conventional partitioning symbols approach: once the symbol sequence is broken into smaller intervals, there is no going back since the data dependencies inside the entropy coders are already broken. There is no flexibility in the number of sequences after encoding, thus the tradeoff between compression rate and parallelism is fixed once the encoding is done. 

We propose an alternative approach, named Recoil, that offers such flexibility. Instead of encoding the symbols with multiple independent encoders, we first use a single group of interleaved rANS encoders to produce a single rANS bitstream. We then split the bitstream by recording the intermediate rANS encoder states before the split point and the symbol indices from which the states are taken. We show that this allows decoding to start at an intermediate position, through a synchronization process. We only pick the intermediate states at renormalization points, as they have a small upper bound and can be represented in fewer bits, reducing storage overhead. We then propose an efficient way to store all this metadata. As the bitstream is independent of the split metadata, the context server only needs to encode once, taking into account the maximum parallelism it intends to support. Then, Recoil enables splits to be combined by a very lightweight process before transmitting to the decoder, since we do not actually divide up the bitstream, but instead record metadata around the split point. 

The main contributions of Recoil are the following:
\begin{itemize}
    \item Recoil allows omitting unnecessary metadata based on the decoder's parallelism capability before transmission, significantly improving the compression rate and saving bandwidth on both sides. We achieved a maximum -14.12\% compression rate overhead reduction in our evaluation.
    \item Recoil makes little tradeoff in decompression speed. Experiments show that Recoil performs comparably in decoding throughput to the conventional method (90+ GB/s on a Turing GPU), outperforming various other ANS-based codecs. 
    \item Recoil is highly compatible with existing rANS decoders, allowing easy integration into standardized codecs. This is because Recoil does not actually modify the rANS bitstream, but instead works on independent metadata.
    \vspace{-1mm}
\end{itemize}

\section{Preliminary}

\subsection{The Range-variant of Asymmetric Numeral Systems (rANS)}
The rANS is a widely-used variant of the modern entropy coder, Asymmetric Numeral Systems. It encodes and decodes each symbol according to a probability density function (PDF) and the corresponding cumulative distribution function (CDF). The PDF and CDF can be either statically pre-computed or adaptively generated. The coding works with a single state, with which the sequence of symbols is represented. Encoding and decoding are performed by a transformation of the state.

\begin{definition}
  Let $S$ be the symbol set, $s$ denote the sequence of input symbols, $\forall i s_i \in S$. Let $f(t)$ and $F(t)$ denote the quantized PDF and CDF of the symbol set $S$. Both $f(t)$ and $F(t)$ are quantized to the range $[0, 2^n]$. Let the coder state after encoding / before decoding symbol $s_i$ be $x_i$. 
  
  Then, the encoding of symbol $s_i$ is formulated as:
\vspace{-0.7mm}
  \begin{equation}
    \label{eqn:rans_enc}
    x_i = 2^n \left\lfloor \frac{x_{i-1}}{f(s_i)} \right\rfloor + F(s_i) + (x_{i-1} \bmod f(s_i))
  \end{equation}

  The decoding of symbol $s_i$ is formulated as:
\vspace{-0.7mm}
  \begin{equation}
    \label{eqn:rans_dec}
    \begin{cases}
      s_i &= t \text{ s.t. } F(t) \leq x_i \bmod 2^n < F(t+1) \\
      x_{i-1} &= f(s) \left\lfloor \frac{x_i}{2^n} \right\rfloor - F(s) + (x_i \bmod 2^n)
    \end{cases}
  \end{equation}

  The division and remainder operations regarding $2^n$ are often implemented with bitwise operations.
\end{definition}

Intuitively, rANS codes a symbol with a smaller bit length when $f(s_i)$ is large, and vice versa, as shown in Equation \ref{eqn:rans_enc}. 

Encoding and decoding in rANS are symmetrical; encoding $s_i$ transforms state $x_{i-1}$ to $x_i$ while decoding gets $s_i$ back by restoring $x_i$  to $x_{i-1}$. Thus, rANS works like a stack: if the symbols are encoded from $s_0$ to $s_n$, the decoder retrieves them in reverse, $s_n$ to $s_0$. While in some implementations, the coder buffers the symbols first and then encodes them in reverse order (so that the decoder retrieves them in forward order), we do not consider this for simplicity.

\textbf{Renormalization.} To avoid handling an unmanageably large state, renormalization is often employed. During encoding, if the state overflows a calculated upper bound, its lower bits are written to a bitstream. The bitstream is read from during decoding when the state underflows a given lower bound.

\begin{definition}
  Let $L = k 2^n, k \in \mathbb{Z}^{+}$ denote the Renormalization Lower Bound. Let $b$ be the number of bits the encoder writes to / decoder reads from the bitstream once, $B$ denote the bitstream, and $p$ denotes the current offset in the bitstream that the codec is writing to / reading from. 
  
  The renormalization during encoding is formulated as:
  \begin{equation}
    \label{eqn:renorm_enc}
    \begin{cases}
      x_i &\coloneqq \left\lfloor \frac{x_i}{2^b} \right\rfloor \\
      B_p &\coloneqq x_i \bmod 2^b \\
      p &\coloneqq p + 1
    \end{cases}
    \text{ while } x_i \geq \frac{2^b}{2^n} L f(s_{i+1})
  \end{equation}

  The renormalization during decoding is formulated as:
  \begin{equation}
    \label{eqn:renorm_dec}
    x_i \coloneqq x_i 2^b + B_p,\,
    p \coloneqq p - 1
    \text{ while } x_i < L
  \end{equation}

  Similarly, this is often implemented with bitwise operations.
\end{definition}

\subsection{Interleaved rANS}
\label{ssec:interleaved_rans}

\begin{figure}[t]
  \centering
  \includegraphics[width=0.9\columnwidth]{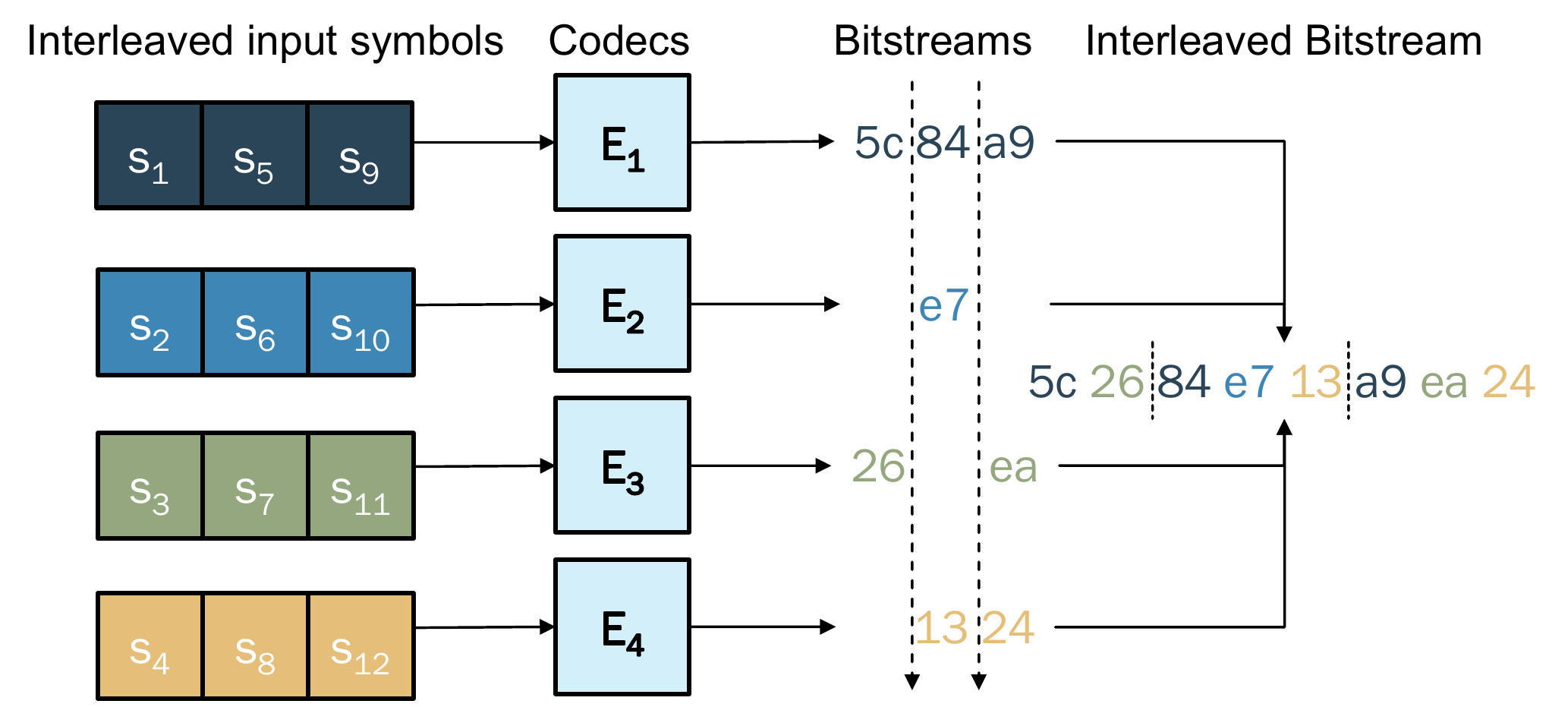}
  \caption{Illustration of interleaved rANS encoding; decoding is done similarly but in reverse.}
  \label{fig:interleaved_rans}
  \vspace{-3mm}
\end{figure}

The rANS decoding algorithm is highly memory-bound. Specifically, the symbol lookup process in Equation \ref{eqn:rans_dec} often consists of an array search or a LUT lookup. Besides, each decoder state $x_i$ depends on the previous state $x_{i+1}$, limiting Instruction Level Parallelism (ILP). Interleaved rANS \cite{InterleavedRANS} is proposed to overcome these limitations. 

An illustration of 4-way interleaved rANS encoding is shown in Figure \ref{fig:interleaved_rans}. In this example, the symbols are processed in groups of 4. Each symbol in the group is encoded by one encoder. After encoding a group, some encoder states overflow, producing renormalization outputs, interleaved into a single bitstream, in the order of increasing encoder ID. The decoding is similarly in reverse: decoders that need to renormalize read from the bitstream in decreasing decoder ID order, before decoding the symbols.

Interleaved rANS shows boosted throughput, because using multiple coders mitigates the data dependency and is thus ILP- and SIMD-friendly. However, interleaved rANS fails to scale over multiple cores. Because the renormalization output must be packed into the bitstream after encoding each group, all threads must stop at a synchronization point to obtain the correct offsets they should write to. Similarly, the offset each decoder reads from during decoding depends on the underflow flag of all other decoders, requiring synchronization. This synchronization barrier is trivial in a core but heavy across multiple cores, and therefore forbids massive scaling over multi-core CPUs and GPUs.
\vspace{-2mm}
\subsection{Conventional "Partitioning Symbols" Approach}

\begin{figure}[t]
  \centering
  \includegraphics[width=0.9\columnwidth]{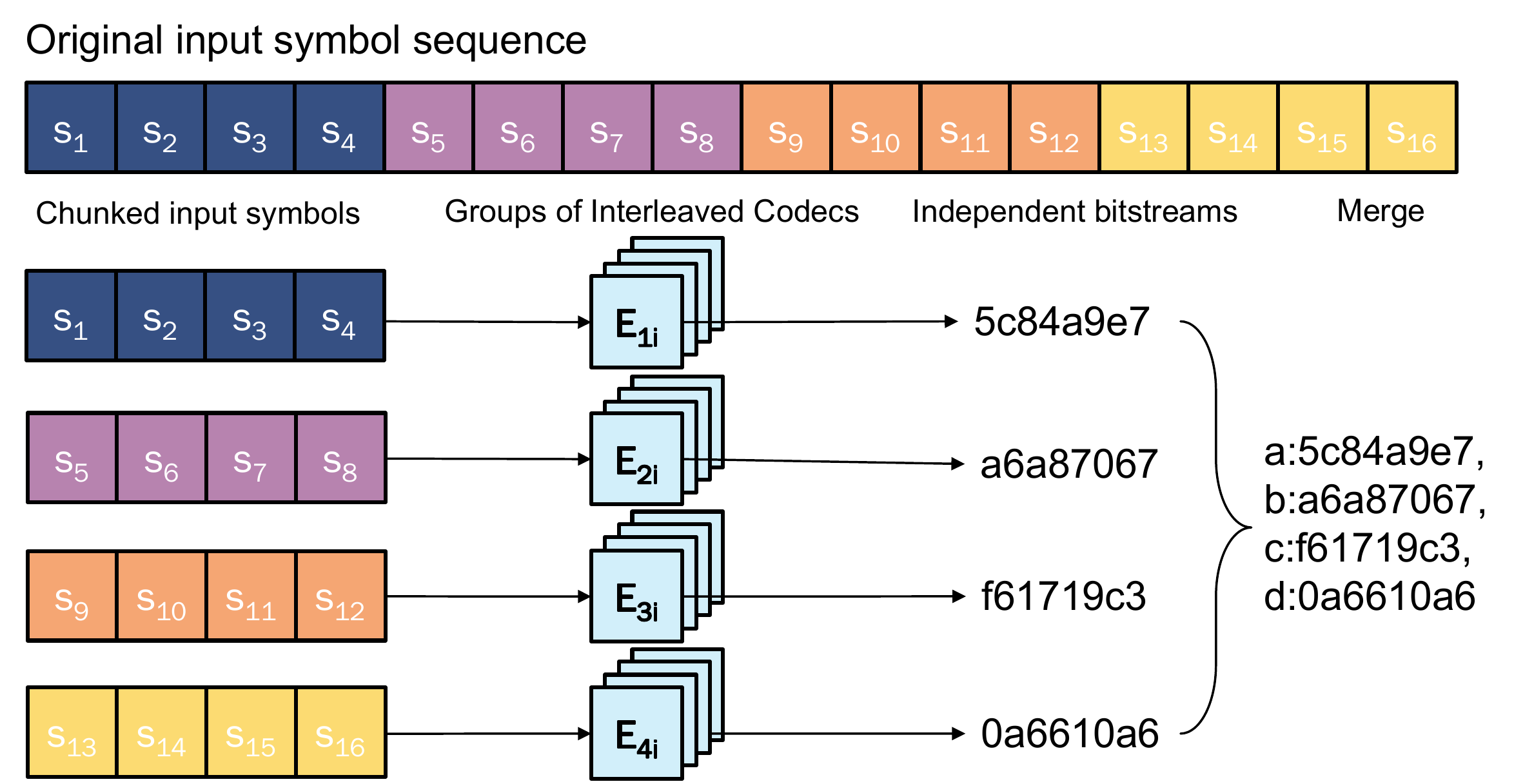}
  \caption{Illustration of the conventional partitioning symbols approach.}
  \label{fig:conv_approach}
  \vspace{-4mm}
\end{figure}

\begin{figure}[t]
  \centering
  \includegraphics[width=0.9\columnwidth]{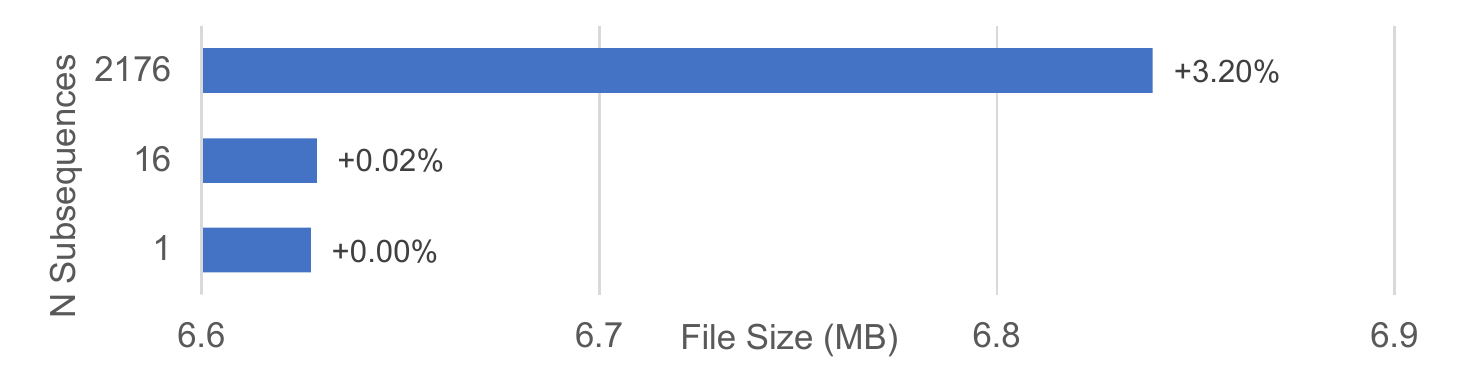}
  \caption{Relationship of the compressed file size and the number of symbol sub-sequences, using the conventional partitioning symbols approach. Evaluated on the first 10 Megabytes of \textbf{enwik9} \cite{enwik9}, using a static distribution quantized to $2^{11}$. The base codec is 32-way interleaved. }
  \label{fig:conv_approach_overhead}
  \vspace{-4mm}
\end{figure}

The conventional approach partitions the input symbol sequence into smaller sub-sequences to achieve further high-throughput rANS decoding on multi-core CPU and GPU systems. The sub-sequences are encoded and decoded by rANS coders completely independent from each other, allowing them to execute in different threads and processors. These coders can be single encoders or decoders, or (in most implementations) groups of interleaved rANS encoders or decoders running on the same CPU core / GPU warp. This process is illustrated in Figure \ref{fig:conv_approach}. This method produces multiple independent bitstreams, which are often merged by simple concatenation and maintaining an offset table to locate the start of each sub-bitstream. The variants of this method, including those built for other ANS variations, are implemented in many previous works, such as DietGPU \cite{DietGPU} and \cite{GST, FPGA_tANS}.

However, partitioning the symbol sequence can come with a huge cost. We evaluated its impact using the partitioning symbols approach and a varying number of sub-sequences. We evaluated 16 sub-sequences, a typical core count of a high-end workstation CPU; and 2176 sub-sequences, the number of threads required to fully utilize a high-end GPU (RTX 2080 Ti). As Figure \ref{fig:conv_approach_overhead} shows, more symbol sub-sequences negatively impacts the compression rate. Especially, the 2176-sub-sequence variation, intended for high-end GPUs, greatly impacts file size. This overhead originates from the initial setup cost of rANS codecs, the final states, etc. Moreover, different hardware has different optimal sub-sequence numbers depending on parallelism capacities. The paradox is that it is impossible to prepare optimal variations of compressed data for every hardware; conversely, if the content delivery server serves the maxed-out variation to all decoders, it creates unnecessary data transfer for those that cannot utilize the max level of parallelism.

The root cause of this problem originates from partitioning symbols, which is irreversible. This breaks the data dependency, which the rANS codec will otherwise establish to achieve efficient coding. Therefore, our approach does not break this dependency chain, but instead uses metadata to provide the missing information, which enables decoding to start from intermediate positions. Our approach does not suffer a similar paradox since more metadata is only sent when the decoder has a larger parallel capacity.
\vspace{-2mm}
\subsection{Other Related Work}

\begin{figure*}[t]
  \centering
  \includegraphics[width=0.75\linewidth]{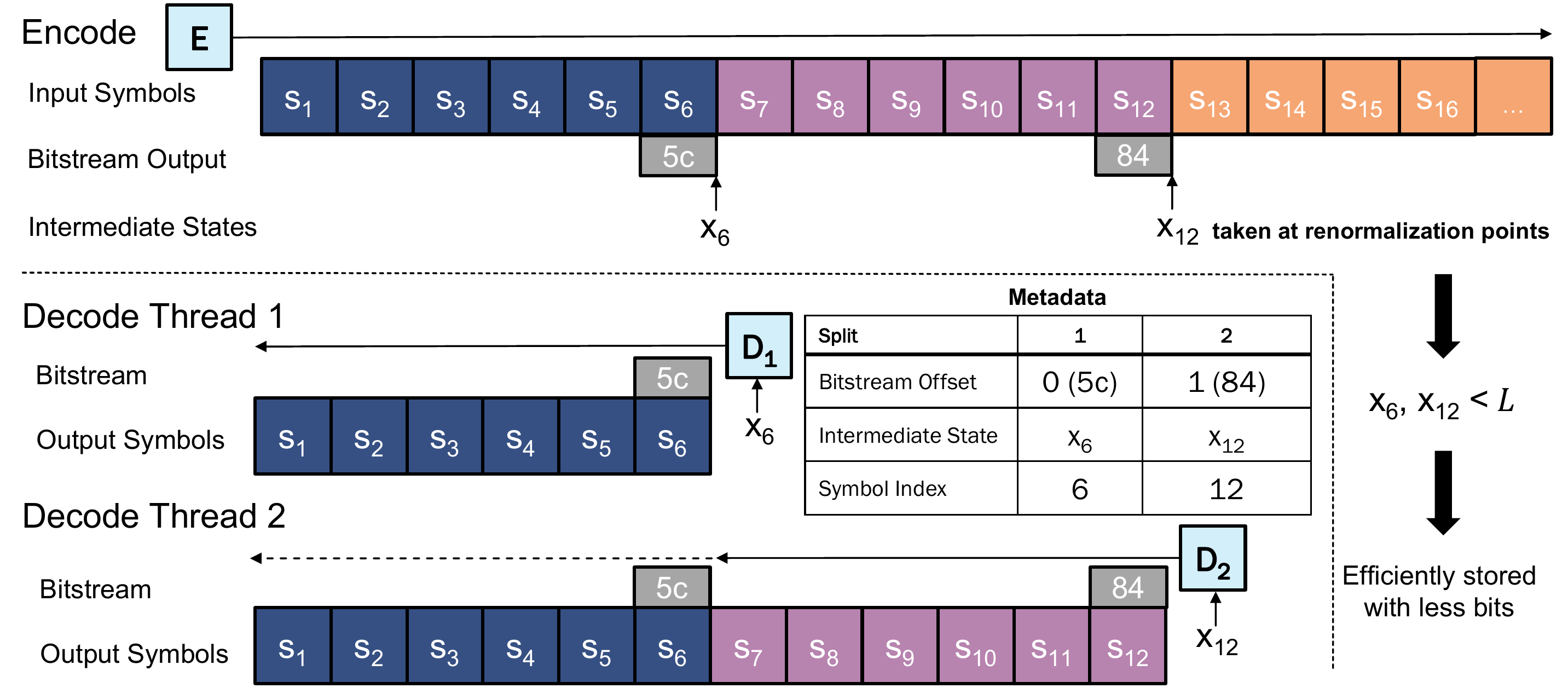}
  \caption{A simple proof of concept of rANS intermediate decodability and decode-adaptive scaling on non-interleaved rANS.}
  \label{fig:motivation}
  \vspace{-2mm}
\end{figure*}

Previous works showcased massively parallel decoding of various compression algorithms \cite{multi_huffman,NVJPEG,NVCOMP,CULZSS,GPU_bz2,GPU_LZ77,GPU_LZW,ndzip-gpu}. In particular, \textbf{multians} \cite{multians} utilizes the fact that tANS (table-variant) tend to self-synchronize even if decoding starts with an incorrect state. The tANS states usually have a limited range (a range of 1024 was used), encouraging self-synchronization. These approaches allow certain encoded bitstreams of corresponding serial encoders to be decoded in parallel, without any extra metadata or file size overhead, seemingly suggesting no harm to compression rate. 

However, the smaller state range limits the maximum probability quantization level $n$. This disables more fine-grained quantization often required in image and video coding. In addition, tANS requires the symbol probability distributions to be pre-computed into decoding tables, which is unsuitable for adaptive coding. While tANS theoretically requires less computation than rANS as its decoding is implemented as a pure LUT, the design of multians \cite{multians} shows a less cache-friendly memory access pattern and large self-synchronization overhead. Due to these reasons, although against instinct, our rANS-based approach greatly outperforms it in decoding speed, as shown in Section \ref{ssec:exp-throughput}.

\section{Motivation}

This section presents the main motivations we considered while designing Recoil, a parallel rANS decoding approach with decoder-adaptive scalability. In particular, we do not partition the uncompressed symbol sequence into sub-sequences before encoding. Instead, our design uses a single group of interleaved rANS encoders, then adds metadata to enable parallel decoding with multiple groups of interleaved rANS decoders. Admittedly, this design trades off encoder throughput; however, there are many use cases where a lower entropy encoder throughput is acceptable, and decoding throughput is the main concern, such as content delivery servers that our design mainly targets. This tradeoff brings greater flexibility in adaptively scaling according to the decoder.

For simplicity of explanation, we demonstrate the method using a bitstream encoded with a single non-interleaved rANS encoder in this section, as shown in Figure \ref{fig:motivation}. In Section \ref{sec:recoil}, we discuss how this method is extended to interleaved rANS bitstreams.

\subsection{Decodability of rANS from intermediate positions}

The rANS bitstream must be decoded from the end to the start because only the final coder state is preserved by being explicitly transmitted along with the bitstream. Let this final state be denoted by $x_n$. The decoder starts with a full bitstream and $x_n$, then iterates to derive $x_{n-1}$, $x_{n-2}$, \ldots, decoding the symbols $s_n$, $s_{n-1}$, $s_{n-2}$, \ldots, and consuming the bitstream during renormalization. 

However, as Equation \ref{eqn:rans_dec} and \ref{eqn:renorm_dec} shows, an iteration of decoding to derive $x_{i-1}$ only depends on two non-static parameters: (1) the previous state $x_i$; (2) the current bitstream offset $p$; assuming a static probability distribution. Since rANS is symmetric and encoder and decoder states are equivalent, it is possible to record these intermediate states and bitstream offsets during encoding, and share this information with the decoder. This enables multiple starting points for decoding. These decoders are completely independent of each other since they do not share either states or bitstream starting offsets, allowing rANS decoding to be scaled over multiple cores. 

For example, as shown in Figure \ref{fig:motivation}, if we record the intermediate state $x_6$, and the starting bitstream offset $0$, it allows decoder thread 1 to start decoding with the state $x_6$. After renormalizing with \textit{5c}, decoder derives the symbols $s_6$, $s_5$, \ldots, $s_1$. Similarly, for state $x_{12}$ and bitstream offset $1$, $s_{12}$ to $s_1$ can be decoded in thread 2.

In practice, we also share the corresponding symbol index at the split point. For example, we would also share that the intermediate state $x_6$ is taken at the index $6$. All metadata we share with the decoder is shown in the metadata table inside Figure \ref{fig:motivation}. This has the following practical advantages: (1) it offers thread 2 an easy way to stop at $s_7$ by simply using a counter. Without this index, each thread does not know how many symbols it must decode, and can only determine the stop points indirectly by checking if the current decoder state and bitstream offset match the start point of another thread; (2) it allows the threads to write to a common buffer of a pre-determined size, rather than dynamically allocating space as needed and copying to one buffer at the end; (3) it allows the use of adaptive coding, in which the probability distribution used in every iteration is dynamic, determined using symbol index as a key in many image codecs that use hyperprior-based context \cite{hyperprior,minnen2018,lin2023}. 

While this seems like a lot of extra information to transmit, we propose a method to efficiently store them in Section \ref{ssec:metadata_store}, so that the burden on compression rate does not exceed, and even outperforms the conventional partitioning symbols approach. In summary, the conventional method partitions the input symbols before encoding, while our method encodes them first and then records metadata that enables parallel decoding, allowing more flexibility.

\subsection{Bounded Intermediate States at Renormalization Points}

It may seem ideal to place split points at where the symbols would be partitioned into equal splits, to balance the workload of threads. However, storing the intermediate states requires many bits (each state can be up to 32 bits in our implementation). Fortunately, we observe that this overhead can be greatly reduced if we only allow splitting at renormalization points.

\begin{lemma}
If an state $x_i$ overflows $\left(x_i \geq \frac{2^b}{2^n} L f(s_{i+1})\right)$ in encoding, then after renormalizing once, $x_i < L$.   
\end{lemma}

\begin{proof}
\renewcommand{\qedsymbol}{}
Recall that $L = k 2^n, k \in \mathbb{Z}^{+}$. The previous state $x_{i-1}$ must be fully renormalized before being used in the next iteration. Thus:
\vspace{-0.7mm}
\begin{equation}
\begin{aligned}
  &x_{i-1} < \frac{2^b}{2^n} L f(s_i) \\
  \Leftrightarrow & \left\lfloor \frac{x_{i-1}}{f(s_i)} \right\rfloor < \frac{2^b}{2^n} L \\
  \Leftrightarrow & \left\lfloor \frac{x_{i-1}}{f(s_i)} \right\rfloor + \underbrace{\frac{F(s_i + (x_i \bmod f(s_i)))}{2^n}}_{< 1} < \frac{2^b}{2^n} L \\
  \Leftrightarrow & \underbrace{2^n \left\lfloor \frac{x_{i-1}}{f(s_i)} \right\rfloor + F(s_i + (x_i \bmod f(s_i)))}_{x_i} < 2^b L
\end{aligned}
\end{equation}

If $x_i$ is renormalized, it is divided by $2^b$. Let $x_i '$ denote the state after renormalization.
\vspace{-0.5mm}
\begin{equation}
x_i' = \left\lfloor \frac{x_i}{2^b} \right\rfloor < L
\vspace{-5mm}
\end{equation}

\end{proof}

In our implementation, $L$ is chosen to be $2^{16}$, so these bounded states can be safely represented in 16-bit numbers, reducing the overhead by half. Renormalization points are where bitstreams are written during encoding. For example, in Figure \ref{fig:motivation}, $x_6$ and $x_{12}$ are taken from renormalization points (which are identified because the encoder produced bitstream outputs there); thus, they are bounded by $L$ and can be represented with fewer bits. 

Intuitively, instead of partitioning the symbols, our method can be seen as splitting the encoded bitstream into sub-bitstreams. We then attempt to balance the number of symbols in each sub-bitstream, so that the workload is evenly distributed.

\subsection{Decoder-Adaptive Scalability}

Recall that the conventional partitioning symbols approach lacks flexibility in decoding scalability: once the bitstream is generated, the number of symbol sub-sequences is fixed, so that decoders with less parallel capacities suffer from a worsened compression rate. In contrast, achieving this flexibility is trivial with our design.

As shown in the metadata table in Figure \ref{fig:motivation}, metadata for thread 2 carries the necessary information needed for that thread to start decoding intermediately from symbol index 12. When no combining of splits happens, thread 2 stops decoding at $s_7$, because $s_1$ to $s_6$ are handled by thread 1. However, nothing prevents thread 2 from continuing decoding beyond that point; it naturally carries all the information required. In other words, if 2-thread-parallelism for the interval of $s_1$ to $s_{12}$ is not needed, we can safely drop the metadata for thread 1; this only removes the ability for decoding to start at symbol index 6. The combined split now appears as a single 12-symbol split to the decoder.

Therefore, combining splits is trivial, since it only requires removing the metadata in a way that combines the splits into bigger ones with close symbol counts, so that the workload is still balanced. Suppose the initial encoding produced $N$ splits metadata and the decoder only wants $M$, $M < N$; we could use a heuristic, but usually simply sending every other $\left\lceil \frac{N}{M} \right\rceil$ split metadata is good enough. As this process is very lightweight and does not require re-encoding of the source symbols, it can be done in real time by the content delivery server before data transmission to the decoder. We consider the use case, where the client requests content, and also attaches its parallel capacity inside the request header; the server receives the request, shrinks down the metadata in real-time, and serves the bitstream and the shrunk metadata to the decoder. No compression rate is wasted to provide unnecessary parallelism.

\section{The Recoil Decoder}
\label{sec:recoil}

\begin{figure}[t]
  \centering
  \includegraphics[width=0.9\columnwidth]{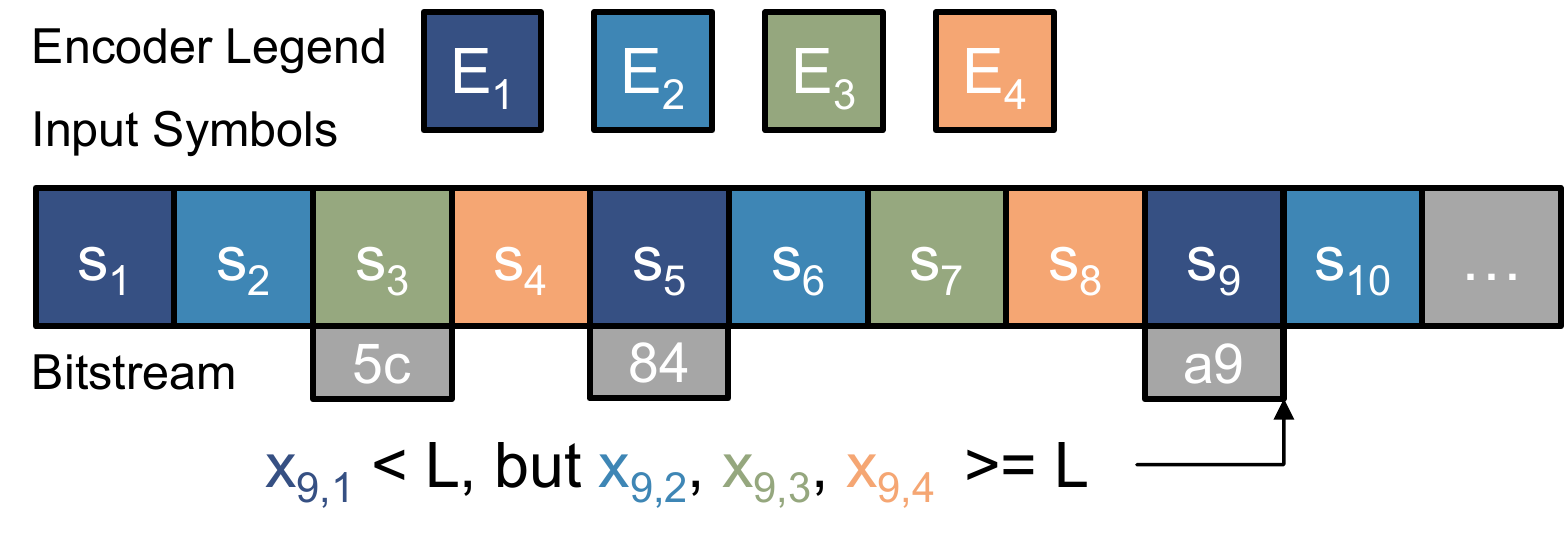}
  \vspace{-3mm}
  \caption{Illustration of that recording all intermediate states at a single position for interleaved rANS results in more storage overhead than expected.}
  \label{fig:extend_interleaved_pitfall}
  \vspace{-4mm}
\end{figure}

\begin{figure*}[t]
  \centering
  \includegraphics[width=0.85\linewidth]{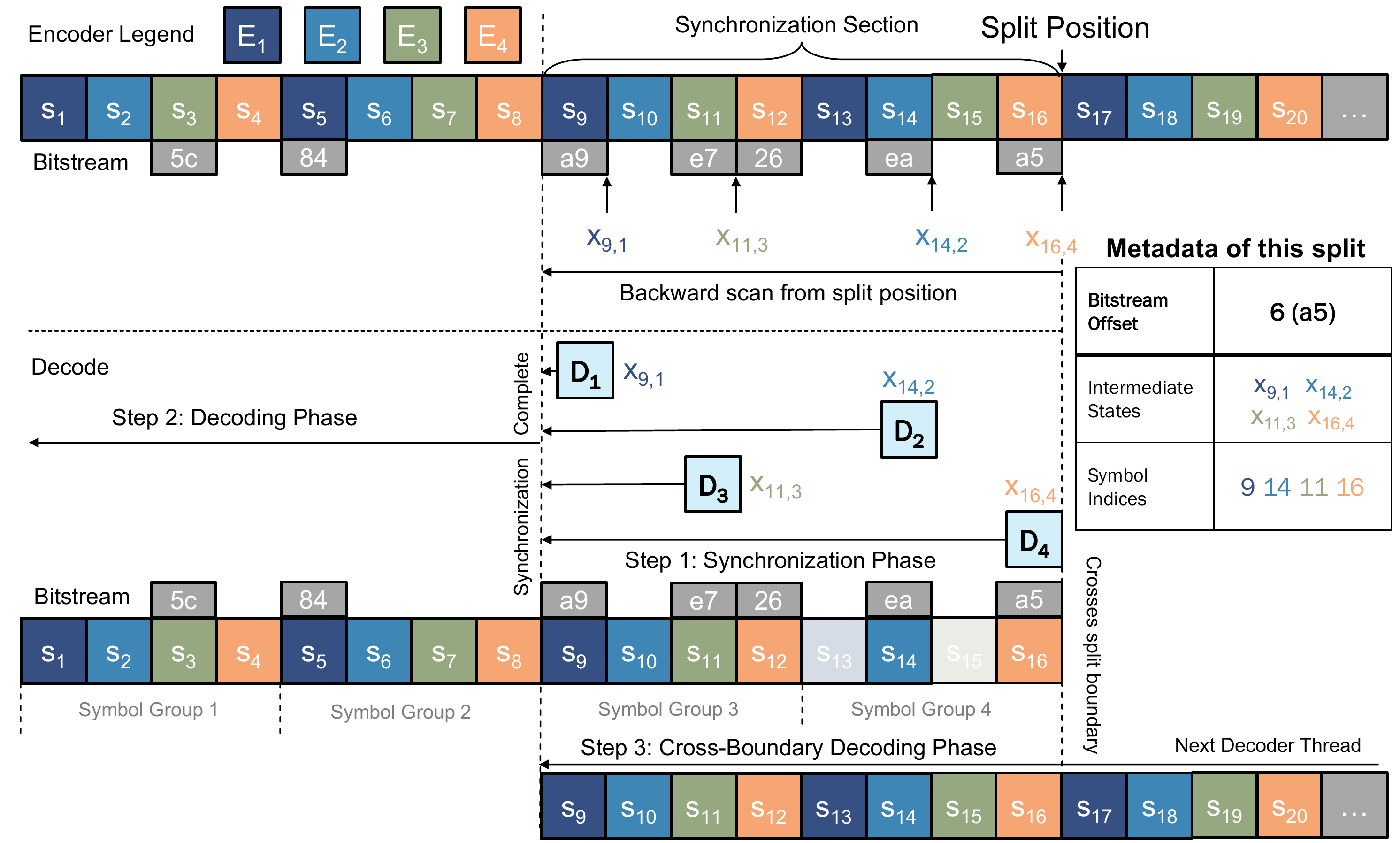}
  \caption{Illustration of Recoil, showing the encoding and decoding procedure of a split.}
  \label{fig:recoil}
  \vspace{-4mm}
\end{figure*}

In the previous section, we explored a potential parallel decoder design, based on non-interleaved rANS codecs. While this approach provides scalability over multi-core, it fails to scale well over SIMD. On the other hand, interleaved rANS \cite{InterleavedRANS} provides such scalability.

In this section, we present Recoil, a massively parallel interleaved rANS decoder with decoder-adaptive scalability that engineers the previous motivations to work over interleaved rANS bitstreams. 

\subsection{Extending to Interleaved rANS}
As shown in Figure \ref{fig:extend_interleaved_pitfall}, recording intermediate states at a single point no longer works with interleaved rANS: if the split point is placed at bitstream offset 2 (\textit{a9}), the intermediate states of the 4 encoders are recorded at this point. However, the only interleaved encoder renormalized here is $E_1$. The other encoders, $E_2$ to $E_4$, encoded other symbols since their last renormalizations. Their intermediate state is no longer below the lower bound and cannot be represented in the smallest number of bits. Therefore, we instead take the intermediate states from multiple close positions during encoding, then decode the bitstream in 3 phases. An illustration is shown in Figure \ref{fig:recoil}. 

During \textbf{encoding}, for a split position, we perform a backward scan to find the last renormalization points of each encoder. For example, in Figure \ref{fig:recoil}, the split position is at $s_{16}$, or bitstream offset 6 (\textit{a5}). Let $x_{i,j}$ denote the intermediate state of interleaved rANS encoder $E_j$ at symbol index $i$. Since the bitstream output \textit{a5} is produced by $E_4$ during renormalization, we record $x_{16,4}$. Then, backward scan finds $E_2$ producing \textit{ea}, so $x_{14,2}$ is recorded. Next, although $E_4$ also renormalized at $s_{12}$ producing \textit{26}, we do not record it as we ignore the previous renormalizations of each encoder. Similarly, we record $x_{11,3}$ and $x_{9,1}$. Since all 4 encoder states have been recorded at index 9, the backward scan terminates. We call the symbol interval $s_9$ to $s_{16}$ the Synchronization Section for this split.

The \textbf{decoding} is done in 3 phases: (1) Synchronization Phase to synchronize and recover the correct intermediate decoder states; (2) Decoding Phase in which normal interleaved rANS decoding is performed; and (3) Cross-Boundary Decoding Phase in which the Synchronization Section of the previous split is decoded.

\subsubsection{\textbf{Synchronization Phase}}

The decoder starts at $s_{16}$, but only the interleaved decoder $D_4$ is initialized here as the intermediate states for $D_1$ to $D_3$ are still unknown. The decoding of $s_{15}$ is skipped since the corresponding interleaved decoder $D_3$ is not initialized. Similarly, we initialize $D_2$ and decode $s_{14}$, and skip $s_{13}$ since $D_1$ is not ready. $s_{12}$ can be decoded correctly since $D_4$ is initialized. We initialize $D_3$, decode $s_{11}$, decode $s_{10}$ with $D_2$, and initialize $D_1$ at index 9. At this point, all 4 decoders are fully initialized, and their states are synchronized through the decoding process. 

As explained in Section \ref{ssec:interleaved_rans}, the interleaved bitstream must be read in a specific order by the interleaved decoders; if the read offset is misaligned, all subsequent decoding results will be incorrect. Recoil ensures the correctness of the read offset even with the absence of some interleaved decoders. This is because we always take the intermediate states from the last renormalization points during encoding. Therefore, interleaved decoders are always initialized immediately before the first time they read the bitstream. Absent interleaved decoders will not be reading the bitstream anyway, thus a correct read offset is always maintained. However, the decoded symbol sequence during this phase is incomplete. For example, in Figure \ref{fig:recoil}, symbols $s_{13}$ and $s_{15}$ are missing. The symbols produced in this stage are merely a side effect of the synchronization, so the decoding results in this phase are discarded.

\subsubsection{\textbf{Decoding Phase}}
After completing synchronization, normal rANS decoding can be performed starting from $s_8$, by the standard interleaved rANS decoding algorithm shown in Figure \ref{fig:interleaved_rans}. This decoding phase ends at the split position boundary.

\subsubsection{\textbf{Cross-Boundary Decoding Phase}}
The next thread handles decoding the next split containing the interval of $s_{17}$ to some later index. After decoding $s_{17}$, this thread finishes the Decoding Phase and now reaches the split boundary at bitstream offset 6 (\textit{a5}). Inherently, it carries all the correct intermediate decoder states required to continue decoding $s_{16}$ and forth. Therefore, this decoder crosses the split boundary, and decodes the symbols in the Synchronization Section of the previous split: $s_9$ to $s_{16}$. It terminates at the synchronization completion point, $s_9$.

\vspace{\baselineskip}

This design reflects the three motivations of Recoil: (1) \textbf{Decodability of rANS from intermediate positions} is achieved through the Synchronization Section and Synchronization Phase; (2) \textbf{Bounded intermediate states at renormalization points} are chosen by the backward scan during encoding; (3) \textbf{Decoder-adaptive scalability} is still trivial to achieve since (1) is maintained, so splits can still be combined by eliminating extra metadata entries.

\subsection{Distributing the workload}

The split points must still be picked carefully to balance the workload. The Synchronization Section adds complexity to this: larger Synchronization Sections bring more runtime overhead. Therefore, the bitstream must be split in a way ensuring both even workloads and less synchronization. We use a simple heuristic to achieve this:

\begin{definition}
  Suppose the full bitstream contains $N$ symbols, and $M$ splits are to be made. For any given split point, let $t$ denote the number of symbols encoded in the sub-bitstream that the previous and the current split points represent (including the Synchronization Section) and $t_s$ denote the number of symbols in the Synchronization Section. We optimize for the minimum of the following function:
\vspace{-1mm}
  \begin{equation}
    H(t, t_s) = \left | t - T \right | + \left | t - t_s - T \right | \text { where } T = \left\lceil \frac{N}{M} \right\rceil
  \end{equation}

  In other words, $T$ represents the expected number of symbols per split; we split at the position that makes the symbol count, including and excluding the Synchronization Section, as close to average as possible.
\end{definition}

Ideally, this heuristic produces many small splits with finely balanced workloads. Suppose $N$ splits are produced. Therefore, to combine these small splits into $M$ more coarse-grained ones, it is only necessary to pick the metadata by every other $\left\lceil \frac{N}{M} \right\rceil$ and send them to the client, achieving decoder-adaptive scaling in real-time.

\subsection{Efficient Metadata Storage}
\label{ssec:metadata_store}

\begin{table}[t]
  \caption{Metadata for the split point shown in Figure \ref{fig:recoil}. We only store the differences from expectations.}
  \vspace{-2mm}
  \label{tab:metadata_split_point}

  \resizebox{7.5cm}{!}{\begin{tabularx}{\columnwidth}{|>{\centering\arraybackslash}X||c|c|c|} 
    \hline
    Split Point Metadata & Expected & Actual & \textbf{Difference} \\
    \hline\hline
    Bitstream Offset & 5 & 6 & \textbf{1} \\
    \hline
    Max Symbol Group ID & 5 & 4 & \textbf{-1} \\
    \hline
   \end{tabularx}}
   \vspace{-3mm}
\end{table}

\begin{table}[t]
  \caption{Codec metadata for the split shown in Figure \ref{fig:recoil}. The intermediate states are stored as-is, while only the differences between the Symbol Group IDs and the Anchor are stored.}
   \vspace{-2mm}
  \label{tab:metadata_codec}

   \resizebox{7.5cm}{!}{\begin{tabularx}{\columnwidth}{|>{\centering\arraybackslash}X||c|c|c|c|} 
    \hline
    Codec Metadata & $E_1$ / $D_1$ & $E_2$ / $D_2$ & $E_3$ / $D_3$ & $E_4$ / $D_4$ \\
    \hline\hline
    \textbf{Intermediate States} & $\mathbf{x_{9,1}}$ & $\mathbf{x_{14,2}}$ & $\mathbf{x_{11,3}}$ & $\mathbf{x_{16,4}}$ \\
    \hline\hline
    Symbol Indices & 9 & 14 & 11 & 16 \\
    \hline
    Symbol Group IDs & 3 & 4 & 3 & 4 \\
    \hline
    Max (Anchor) & \multicolumn{4}{c|}{4} \\
    \hline
    \textbf{Differences} & \textbf{-1} & \textbf{0} & \textbf{-1} & \textbf{0} \\
    \hline
   \end{tabularx}}
   \vspace{-3mm}
\end{table}

Our heuristic ensures that the number of symbols per split will not differ largely from the average. Also, most real-world data has a mostly uniform distribution of entropy: in ideal conditions, equal numbers of symbols are compressed into equal bitstream lengths. While real-world datasets are not always ideal, the actual split points are likely not far away from this expectation. We exploit these facts to store the metadata efficiently.

First, we store the number of splits $M$, the bitstream length $B$, and the number of all symbols $N$ as-is in the metadata header. We then compute the expected sub-bitstream length $E_b$ for each split simply using rounded-up averages: $E_b = \lceil\frac{B}{M}\rceil$. The $i$-th split point is expected to be at bitstream offset $iE_b$. Then, we calculate the differences between the actual positions and these expectations, and only store the differences in the metadata, as Table $\ref{tab:metadata_split_point}$ shows.

We apply a similar approach to the metadata of each interleaved codec. The intermediate states are stored as-is since they are difficult to be encoded further. We do not store the actual symbol indices but only the Symbol Group IDs corresponding to the ones in Figure \ref{fig:recoil}, since it is trivial to convert them back and forth. An example is shown in Table \ref{tab:metadata_codec}: we take the max Symbol Group ID 4, and store it in the form of difference to the expectation, as shown in Table \ref{tab:metadata_split_point}; the expected value here is also simply a rounded-up average.

Next, we use this maximum value as an anchor, then compute and store the differences from all the Symbol Group IDs to it. This is based on the assumption that when one interleaved encoder renormalizes, the other encoders should also renormalize soon. The differences here are guaranteed to be negative or zero since they are compared to the maximum value. We drop the sign bits and store the absolute values only.

We group the difference values into data series and store them in the following format: (1) we first use a value to represent how many maximum bits each element occupies, and minus it by one, which is equal to $\max \lfloor \log_2 (v_i+1) \rfloor - 1$ where $v_i$-s are the individual values in the data series\footnote{Except for zero: we use one bit to represent zeros as well.}; (2) we then encode the series with this number of bits, with an extra sign bit if necessary. 

We allow up to 16-bit unsigned values for the Symbol Group ID differences, and therefore a 4-bit value for the bit count. We group the differences from each split into individual data series. Therefore, the Differences row in Table \ref{tab:metadata_codec} is stored as:
\vspace{-1mm}
\begin{displaymath}
  \underbrace{0000}_{\text{len} = 0 + 1 = 1 \text{ bit}}\underbrace{1}_{-1}\underbrace{0}_{0}\underbrace{1}_{-1}\underbrace{0}_{0}
\end{displaymath}

For the split point metadata in Table \ref{tab:metadata_split_point}, we allow up to 32-bit signed values and group the differences from all the splits into two data series containing all Bitstream Offset and Max Symbol Group ID differences, respectively.

\subsection{Implementation Details}

\begin{table}[t]
  \caption{Recommended rANS parameters for our implementation.}
  \vspace{-2mm}
  \label{tab:rans_params}

   \resizebox{7.5cm}{!}{\begin{tabularx}{\columnwidth}{|c|>{\centering\arraybackslash}X|c|} 
    \hline
    Symbol & Description & Value \\
    \hline\hline
    \texttt{sizeof}($x_i$) & size of rANS states & 32 bits \\
    \hline
    \texttt{sizeof}($s_i$) & size of symbols & 8 or 16 bits \\
    \hline
    $L$ & Renormalization lower bound & $2^{16}$ \\
    \hline
    $b$ & Renormalization output size & 16 bits \\
    \hline
    $n$ & PDF / CDF quantization level & \textit{varying}, max 16 \\
    \hline
    $|E|$ / $|D|$ & Number of interleaved codecs & 32 \\
    \hline
   \end{tabularx}}
   \vspace{-5mm}
\end{table}

We implemented four variations of Recoil decoding: (1) a pure C++ implementation, non-optimized as it is for debugging purposes; (2) an AVX2 implementation; (3) an AVX512 implementation; (4) a CUDA implementation for executing on NVIDIA GPUs. We expect the algorithm to be easily ported to any other SIMD + multi-core architecture, including other GPUs, since we do not rely on any platform-specific feature. We also include a basic encoder for testing purposes. All four implementations are mutually compatible; generated bitstreams by the encoder can be decoded by any of them. Implementations (2) and (3) can be selected based on the target platform's AVX support when decoding on CPU.

Recoil is implemented as a C++20 header-only library and relies heavily on templates to allow customizing almost every parameter. However, for the best performance, we recommend the parameters in Table \ref{tab:rans_params}, and used them throughout the experiments.

We use 32-way interleaved rANS because it performs best for both AVX implementations and naturally fits into a GPU warp. For the AVX2 implementation, we use 8-way 32-bit interleaved decoders in each instruction, and manually unroll four times; for the AVX512 implementation, we use 16 ways in each instruction and unroll twice. We recommend launching one thread per CPU core and not utilizing SMT. For CUDA, we use 128 threads per block operating on four groups of interleaved decoders and use a CUDA library function call\footnote{\texttt{cudaOccupancyMaxActiveBlocksPerMultiprocessor}} to obtain the optimal block count. We build LUTs for the symbol lookup process shown in equation \ref{eqn:rans_dec}. Here we apply a common optimization: if \texttt{sizeof}($s_i$) = 8, and $n \leq 12$, we pack the symbol $s_i$, its quantized probability $f(s_i)$ and quantized CDF $F(s_i)$ into a single 32-bit integer.

Our performance implementations, namely (2) (3) (4), expect that renormalization always completes in one step. The necessary condition is $b \geq n$ \cite{InterleavedRANS}. Therefore, we pick $b = 16$ by default to support probability quantization levels up to $16$, which is more than enough for most applications.
\vspace{-2mm}
\section{Experiments}
\label{sec:experiments}

\subsection{Experiment Setup}

\begin{table}[t]
  \caption{Description of the test datasets. The baseline compressed sizes are of variation (a). 1 KB = 1000 Bytes.}
  \vspace{-3mm}
  \label{tab:datasets}
  \centering
  \resizebox{7.5cm}{!}{\begin{tabular}{|c|c|c|c|}
    \hline
    \multirow{2}{*}{Name} & \multirow{2}{*}{Uncompressed} & \multicolumn{2}{c|}{Compressed [variation (a)]} \\
    & & \multicolumn{1}{c}{$n = 11$} & \multicolumn{1}{c|}{$n = 16$} \\ \hline\hline
    rand\_10 & 10,000 KB & 7,828 KB & 7,657 KB \\ \hline
    rand\_50 & 10,000 KB & 5,357 KB & 4,774 KB \\ \hline
    rand\_100 & 10,000 KB & 4,157 KB & 3,534 KB \\ \hline
    rand\_200 & 10,000 KB & 3,045 KB & 2,317 KB \\ \hline
    rand\_500 & 10,000 KB & 1,395 KB & 886 KB \\ \hline\hline
    dickens \cite{Datasets} & 10,192 KB & 6,268 KB & 5,794 KB \\ \hline
    webster \cite{Datasets} & 41,459 KB & 27,375 KB & 25,832 KB \\ \hline
    enwik8 \cite{enwik9} & 100,000 KB & 66,128 KB & 63,588 KB \\ \hline
    enwik9 \cite{enwik9} & 1,000,000 KB & 672,816 KB & 645,443 KB \\ \hline\hline
    div2k801 \cite{div2k} & 7,209 KB & N/A & 2,093 KB \\ \hline
    div2k803 \cite{div2k} & 7,864 KB & N/A & 3,208 KB \\ \hline
    div2k805 \cite{div2k} & 7,864 KB & N/A & 1,496 KB \\ \hline
  \end{tabular}}
  \vspace{-6mm}
\end{table}

\begin{table*}[h]
  \caption{Differences of compressed file sizes compared to variation (a) baseline in Table \ref{tab:datasets}, n=11.}
  \vspace{-1mm}
  \label{tab:overhead-n-11}
  \centering
  \resizebox{17cm}{!}{\begin{tabular}{|l|l|l|l|l|l|l|l|l|l|l|}
    \hline
    Dataset & \multicolumn{2}{l|}{(b) Conventional Large} & \multicolumn{2}{l|}{(c) Recoil Large} & \multicolumn{2}{l|}{(d) Conventional Small} & \multicolumn{2}{l|}{(e) Recoil Small} & \multicolumn{2}{l|}{(f) multians} \\ \hline
    rand\_10 & +211.44 KB & +2.70\% & +163.67 KB & +2.09\% & +1.47 KB & +0.02\% & +1.12 KB & +0.01\% & +76.43 KB & +0.98\%  \\ \hline
    rand\_50 & +211.37 KB & +3.95\% & +170.35 KB & +3.18\% &  +1.45 KB & +0.03\% & +1.16 KB & +0.02\% & -177.66 KB & -3.32\%  \\ \hline
    rand\_100 & +211.25 KB & +5.08\% & +172.91 KB & +4.16\% & +1.45 KB & +0.03\% & +1.18 KB & +0.03\% & -178.29 KB & -4.29\%  \\ \hline
    rand\_200 & +211.32 KB & +6.94\% & +179.39 KB & +5.89\% & +1.27 KB & +0.04\% & +1.09 KB & +0.04\% & -355.64 KB & -11.68\% \\
    \hline
    rand\_500 & +203.31 KB & +14.57\% & +189.57 KB & +13.59\% & +1.27 KB & +0.09\% & +1.14 KB & +0.08\% & -132.65 KB & -9.51\% \\
    \hline\hline
    dickens & +211.90 KB & +3.38\% & +165.07 KB & +2.63\% & +1.46 KB & +0.02\% & +1.14 KB & +0.02\% & -97.93 KB & -1.56\% \\ \hline
    webster & +211.71 KB & +0.77\% & +165.30 KB & +0.60\% & +1.46 KB & +0.01\% & +1.12 KB & +0.00\% & -119.88 KB & -0.44\%  \\ \hline
    enwik8 & +212.70 KB & +0.32\% & +165.56 KB & +0.25\% & +1.45 KB & +0.00\% & +1.12 KB & +0.00\% & +510.35 KB & +0.77\% \\ \hline
    enwik9 & +213.15 KB & +0.03\% & +166.89 KB & +0.02\% & +1.48 KB & +0.00\% & +1.14 KB & +0.00\% & +3,384.83 KB & +0.50\% \\ \hline
  \end{tabular}}
  \vspace{-2mm}
\end{table*}

\begin{table*}[h]
  \caption{Differences of compressed file sizes compared to variation (a) baseline in Table \ref{tab:datasets}, n=16.}
  \vspace{-1mm}
  \label{tab:overhead-n-16}
  \centering
  \resizebox{17cm}{!}{\begin{tabular}{|l|l|l|l|l|l|l|l|l|l|l|}
    \hline
    Dataset & \multicolumn{2}{l|}{(b) Conventional Large} & \multicolumn{2}{l|}{(c) Recoil Large} & \multicolumn{2}{l|}{(d) Conventional Small} & \multicolumn{2}{l|}{(e) Recoil Small} & \multicolumn{2}{l|}{(f) multians} \\ \hline
    rand\_10 & +211.19 KB & +2.76\% & +163.94 KB & +2.14\% & +1.47 KB & +0.02\% & +1.12 KB & +0.01\% & +200.56 KB & +2.62\%  \\ \hline
    rand\_50 & +210.56 KB & +4.41\% & +171.53 KB & +3.59\% & +1.47 KB & +0.03\% & +1.15 KB & +0.02\% & +337.15 KB & +7.06\%  \\ \hline
    rand\_100 & +211.03 KB & +5.97\% & +172.10 KB & +4.87\% & +1.46 KB & +0.04\% & +1.17 KB & +0.03\% & +358.53 KB & +10.15\%  \\ \hline
    rand\_200 & +208.95 KB & +9.02\% & +180.90 KB & +7.81\% & +1.30 KB & +0.06\% & +1.09 KB & +0.05\% & +372.44 KB & +16.07\% \\
    \hline
    rand\_500 & +208.53 KB & +23.54\% & +190.75 KB & +21.53\% & +1.27 KB & +0.14\% & +1.14 KB & +0.13\% & +376.81 KB & +42.54\% \\
    \hline\hline
    dickens & +211.56 KB & +3.65\% & +164.79 KB & +2.84\% & +1.46 KB & +0.03\% & +1.13 KB & +0.02\% & +312.40 KB & +5.39\% \\ \hline
    webster & +211.23 KB & +0.82\% & +165.03 KB & +0.64\% & +1.46 KB & +0.01\% & +1.12 KB & +0.00\% & +1,206.48 KB & +4.67\%  \\ \hline
    enwik8 & +212.35 KB & +0.33\% & +165.28 KB & +0.26\% & +1.47 KB & +0.00\% & +1.12 KB & +0.00\% & +2,507.60 KB & +3.94\% \\ \hline
    enwik9 & +213.10 KB & +0.03\% & +166.63 KB & +0.03\% & +1.51 KB & +0.00\% & +1.13 KB & +0.00\% & +25,660.70 KB & +3.98\% \\ \hline\hline
    div2k801 & +215.75 KB & +10.31\% & +173.41 KB & +8.28\% & +1.46 KB & +0.07\% & +1.18 KB & +0.06\% & \multicolumn{2}{c|}{N/A} \\ \hline
    div2k803 & +224.23 KB & +6.99\%  & +172.35 KB & +5.37\% & +1.49 KB & +0.05\% & +1.17 KB & +0.04\% & \multicolumn{2}{c|}{N/A} \\ \hline
    div2k805 & +212.47 KB & +14.20\% & +176.58 KB & +11.80\% & +1.49 KB & +0.10\% & +1.23 KB & +0.08\% & \multicolumn{2}{c|}{N/A} \\ \hline
  \end{tabular}}
\end{table*}

We use 10 datasets to evaluate the codecs featuring different sizes and compressibilities, as shown in Table \ref{tab:datasets}. The \texttt{rand\_*} datasets are 10-Megabyte files generated with random exponentially distributed bytes, with $\lambda = 10, 50, 100, 200, 500$ respectively representing different compression rates. The dickens \cite{Datasets}, webster \cite{Datasets}, enwik8 \cite{enwik9} and enwik9 \cite{enwik9} datasets are ASCII text files of various sizes. We model these 7 datasets with static probability distributions generated by symbol statistics and use 8-bit symbol size. The \texttt{div2k*} datasets are high-quality images from the DIV2K validation set \cite{div2k}. We use the mbt2018-mean \cite{minnen2018} lossy codec to transform the source image into intermediate representations of 16-bit symbols. We adaptively model each symbol with different Gaussian distributions using hyperpriors.

We compare Recoil against the following baselines: (A) \textbf{Single-Thread}  or standard 32-way interleaved rANS; (B) the \textbf{Conventional} partitioning symbols approach with underlying 32-way interleaved rANS; and (C) \textbf{multians} \cite{multians}. To achieve the flexibility of parameters required in the experiments, and minimize the implementation differences so that the comparison focuses on the algorithms, we implemented (A) and (B) with the same building blocks that constitute Recoil. We checked that our baseline implementations show comparable throughput with state-of-the-art\footnote{We checked that our implementation of Single-Thread outperforms ryg\_rans \cite{ryg_rans}; the latter is a popular but SSE-4.1-based implementation which fails to unleash the full potential of modern CPUs. We compare our implementation of Conventional on GPU with DietGPU \cite{DietGPU} and it shows similar performance; DietGPU assumes probability quantization levels of 9 to 11 as part of the optimization strategy, which is not enough to perform some of the experiments.} implementations. We use the open-source implementation for multians but modify the state count only for the $n=16$ experiment; however, since it has limited support for adaptive coding, we omit it for the image compression tests.

\subsection{Compression Rate Results}

We first compress the datasets with two probability quantization levels: $n=11$ and $n=16$ and compare the compression rates. We only compress the image datasets with $n=16$ because a higher quantization level is required for 16-bit symbols. For each dataset and $n$, we encode the bitstream into six variations: (a) standard rANS bitstream compressed with \textbf{Single-Thread} serving as compression rate baseline, shown in Table \ref{tab:datasets}; (b) (c) the Large variations (2176 partitions/splits) for massively parallel GPU decoding, encoded with \textbf{Conventional} and \textbf{Recoil} correspondingly; (d) (e) the Small variations (16 partitions/splits) for parallel CPU decoding, (d) re-encoded with \textbf{Conventional}, and (e) converted from the Large (c) variation using \textbf{Recoil} splits combining; (f) tANS bitstream for decoding with \textbf{multians}. 

The differences in the compressed file sizes of variations (b) to (f), compared to the baseline (a), are shown in Table \ref{tab:overhead-n-11} and \ref{tab:overhead-n-16} for $n=11, 16$ respectively. The Small variations show negligible overheads, with max. 0.14\% for Conventional and 0.13\% for Recoil; however, the Large variations impact the compression rates more, as a max. 23.54\% increase in file size was observed for Conventional. In contrast, even without splits combining, Recoil variations already show lower overhead, with only max. 21.53\% and outperforming Conventional in every dataset. It is clear that when the baseline compressed file size is small and more splits are made, the parallelism impacts on compression rate becomes more dominant.

Recall that symbol partitions in Conventional cannot be combined; the content therefore must be either encoded into both Large (b) and Small (d) variations, creating extra storage and encoding overhead on the server, or every client is served the Large (b) variation, creating unnecessary data transfer. Recoil allows real-time conversion from Large (c) to Small (e), which is especially effective when the original compressed file size is small. This reduces max. -23.41\% compression rate overhead compared to sending the Conventional Large (b) variation, and eliminates the need for multiple encoded variations of the same data on the server.

Interestingly, some multians bitstreams at $n=11$ outperformed the baseline in compression rate, as shown in Table \ref{tab:overhead-n-11}. This is likely because the single-threaded baseline consists of 32-way interleaved rANS encoders, while multians use a single tANS encoder. However, this advantage vanished when we set $n=16$, as shown in Table \ref{tab:overhead-n-16}. 

\subsection{Decoding Throughput Results}
\label{ssec:exp-throughput}
\begin{figure*}[t]
  \centering
  \includegraphics[width=0.98\linewidth]{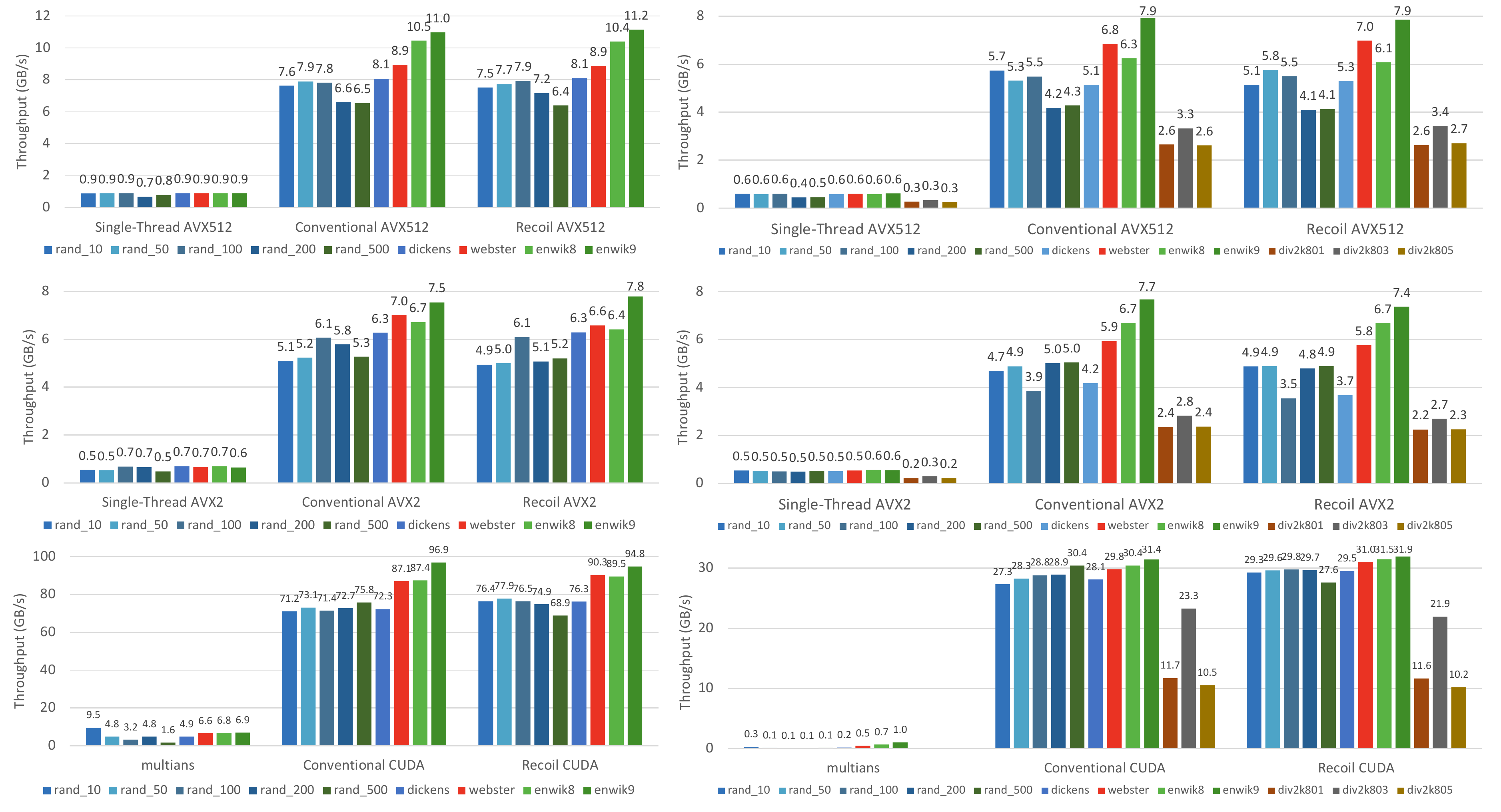}
  \vspace{-4mm}
  \caption{Performance comparison of baselines (A) and (B) with Recoil on CPU (AVX512 and AVX2) and of baseline (C) multians and (B) with Recoil on GPU, with n=11 (Left) and n=16 (Right).}
  \label{fig:perf-all}
  \vspace{-4mm}
\end{figure*}

Next, we test the decoding of the six bitstream variations with the corresponding decoders and measure the throughput. On CPU, we compare the AVX512 and AVX2 implementations of \textbf{Single-Thread} and \textbf{Conventional}, and \textbf{Recoil}, by decoding bitstream variations (a), (d) and (e). On GPU, we compare the CUDA implementations of \textbf{multians}, \textbf{Conventional} and \textbf{Recoil}, by decoding bitstreams (f), (b) and (c). We use Intel Xeon W-3245 (16C) for CPU experiments and NVIDIA GeForce RTX 2080 Ti for GPU experiments. We measure only the CUDA kernel execution for GPU, excluding memory transfer overhead. We average the throughput over 10 runs.

The throughput measurement results on CPU and GPU are shown in Figure \ref{fig:perf-all}. On CPU, both Conventional and Recoil greatly outperformed Single-Thread, showing massive scalability across multiple CPU cores, and using AVX512 yields even higher performance. The decoding throughput of Recoil and Conventional are comparable, reaching max. 11+ GB/s decoding throughput, showing that the extra synchronization and cross-boundary decoding steps of Recoil only bring negligible performance overhead. 

The conclusion on GPU is similar: Recoil and Conventional performed similarly, reaching max. 90+ GB/s decoding throughput. This indicates that the heuristic Recoil uses can distribute work evenly across threads even when there are many splits, and that the synchronization overhead is mostly negligible. 

Both rANS-based variations significantly outperformed multians, the tANS-based parallel decoder, which can even be matched by our CPU-based decoders. Especially when we manually set $n=16$, the self-synchronization overhead of multians is too large that it fails to output symbols at a usable throughput. Therefore, although multians is a zero-storage-overhead parallel tANS decoder, it has strong limitations that ultimately restrict its compression rate in other ways: we could use $n=16$ and still outperform the original multians with $n=11$ in both throughput and compression rate.
\vspace{-2mm}
\section{Conclusion}
\label{sec:conclusion}

In this work, we presented Recoil, a parallel interleaved rANS decoder that adaptively combines splits to match the parallel capacity of each decoder, which avoids wasting compression rate to provide unnecessary parallelism. We proposed to encode the symbols into a single interleaved rANS bitstream, then pick bitstream split points using a heuristic and record metadata to enable parallelism. We observed comparable decoding throughput to the conventional approach, outperforming other ANS decoders, while greatly saving compression rate. 

Although Recoil encoding cannot be done in parallel and encoding throughput is limited, this is often acceptable in content delivery applications, especially those involving high-resolution image, UHD video, VR and AR contents: complex and non-parallel prediction algorithms are widely used in their compression pipelines, in an effort to reduce data rate while maintaining quality, so that Recoil encoding is unlikely the bottleneck of these systems.

As future work, Recoil can be combined with standardized state-of-the-art image and video coding formats. Recoil can be an easy drop-in replacement for the single-threaded interleaved rANS coders: the Recoil metadata can be transmitted separately so that the coding format does not change. This creates massively parallel high-throughput and low-latency content delivery experiences while wasting no transmission overhead on unused parallelism. 

The open-source implementation of Recoil is made public at \href{https://github.com/lin-toto/recoil}{https://github.com/lin-toto/recoil}.
\vspace{-2mm}
\begin{acks}
This work was supported in part by NICT No. 03801, JST PRESTO JPMJPR19M5, JSPS Grant 21K17770, Kenjiro Takayanagi Foundation, and Foundation of Ando laboratory.

The authors would like to thank \textbf{Shota Hirose}, \textbf{Jinming Liu}, \textbf{Ao Luo}, \textbf{Zhongfa Wang}, \textbf{Huanchao Shen}, \textbf{Hiroshi Sasaki}, \textbf{Prof. Keiji Kimura}, \textbf{Chisato Nishikigi}, and \textbf{Takina Inoue} for the inspiring comments and discussions on this work. The project name is inspired by the TV anime series \textbf{"Lycoris Recoil"}, and we gratefully thank them for driving our research with the power of \textit{sakana} and \textit{chinanago}.
\end{acks}
\vspace{-2mm}

\bibliographystyle{ACM-Reference-Format}
\bibliography{acmart}


\end{document}